\documentclass{sig-alternate}
\usepackage{hyperref}
\usepackage{caption}
\usepackage{subcaption}
\usepackage{hhline}

\newtheorem{lemma}{Lemma}

\begin{document}

\conferenceinfo{KDD}{'15 Sydney, Australia}

\title{Virus Propagation in Multiple Profile Networks} 


%
%
%
%
%
\numberofauthors{4} 
%
\author{
%
%
\alignauthor
A. Rapti\\
       \affaddr{University of Patras}\\
       \affaddr{Rio}\\
       \affaddr{Patras, Greece}\\
       \email{arapti@ceid.upatras.gr}
\alignauthor
S. Sioutas \\ 
       \affaddr{Ionian University}\\
       \affaddr{Corfu}\\
       \affaddr{Corfu, Greece}\\
       \email{sioutas@ionio.gr}
\and  
\alignauthor K. Tsichlas\\
       \affaddr{Aristotle University of Thessaloniki}\\
       \affaddr{Thessaloniki}\\
       \affaddr{Greece}\\
       \email{tsichlas@csd.auth.gr}
\alignauthor G. Tzimas\\
       \affaddr{Technological Educational Institute of Western Greece}\\
       \affaddr{Nafpaktos}\\
       \affaddr{Greece}\\
       \email{tzimas@cti.gr}
}

\maketitle
\begin{abstract}

Suppose we have a virus or one competing idea/product that propagates over a multiple profile (e.g., social) network. Can we predict what proportion of the network will actually get "infected" (e.g., spread the idea or buy the competing product), when the nodes of the network appear to have different sensitivity based on their profile? For example, if there are two profiles $\mathcal{A}$ and $\mathcal{B}$ in a network and the nodes of profile $\mathcal{A}$ and profile $\mathcal{B}$ are susceptible to a highly spreading virus with probabilities $\beta_{\mathcal{A}}$ and $\beta_{\mathcal{B}}$ respectively, what percentage of both profiles will actually get infected from the virus at the end? To reverse the question, what are the necessary conditions so that a predefined percentage of the network is infected? We assume that nodes of different profiles can infect one another and we prove that under realistic conditions, apart from the weak profile (great sensitivity), the stronger profile (low sensitivity) will get infected as well. First, we focus on cliques with the goal to provide exact theoretical results as well as to get some intuition as to how a virus affects such a multiple profile network. Then, we move to the theoretical analysis of arbitrary networks. We provide bounds on certain properties of the network based on the probabilities of infection of each node in it when it reaches the steady state. Finally, we provide extensive experimental results that verify our theoretical results and at the same time provide more insight on the problem.
\end{abstract}

\category{H.2.8}{Database management}{Database applications - Data mining}

\keywords{Epidemics, Virus Propagation, Profiles}

\section{Introduction}

Suppose a new game has been published for PS4. The fans of PS4 will start posting tweets on Twitter about this new game and their followers/friends will retweet in turn, provided that they are also fans of PS4. Followers/agents that are not interested in PS4 will be neutral towards this tweet while followers/agents that are fans of XBOX One will be rather indifferent, if not hostile, and will not retweet, at least not to this extent. It is natural to expect that this information about this game will propagate through most of the PS4 fans, but will have a smaller degree of propagation to other groups of agents. This simple example shows a fundamental truth in information propagation over social networks: that not all agents have the same affinity for some particular piece of information. This is true in most of the settings where an idea/rumor/virus propagates subject to particular restrictions over the nodes of a network. This paper is a first attempt (to the best knowledge of the authors) to incorporate the affinity/susceptibility of agents towards a particular idea/rumor/virus under a specific propagation model. 

The general setting considered in this paper consists of a (or many) virus that try to infect as many nodes as possible in a given network. The rules of the game are dictated by the SIS model, which is an epidemiological model where each node can be in any of the two states: susceptible to infection or infected. As implied by our previous discussion, we use this epidemiological model in the general setting of information diffusion, meaning that a virus may as well correspond to a piece of information diffused on a social network subject to the rules posed by the SIS model.

There are many interesting questions arising in this setting. How will a virus propagate over a given network? Can we determine whether all members of the network will be infected or will it spread in a small group in the network and then die out quickly? Similarly, when can we say that such an infectious virus will take off? What are the necessary conditions for the virus to flood the network? What are the necessary conditions related to the characteristics of the network (e.g., mean degree) for a particular virus to infect most (or some percentage of the network)? Finally, what is the case when the nodes have different endurance/sensitivity to the "virus" and have temporary or permanent immunity?

\begin{table}[tb]
\begin{center}
\caption{Theoretical results presented in this paper. The maximum eigenvalue of a square matrix is represented by $\rho$.} \label{tab:res}
\begin{tabular}{|c | c |}
\hline
{\bf Fixed Point} & {\bf Condition}  \\ \hhline{|=|=|}
\textbf{Clique} & \\ \hline
$I_{\mathcal{A}},I_{\mathcal{B}} \rightarrow 0$ & $N\frac{\delta_{\mathcal{A}}\beta_{\mathcal{B}}+\delta_{\mathcal{B}}\beta_{\mathcal{A}}}{\delta_{\mathcal{A}}\delta_{\mathcal{B}}}<1$ \\ \hline
$I_{\mathcal{A}},I_{\mathcal{B}} = cN, c\in\left(\frac{1}{2},1\right)$ & $\beta_{\mathcal{A}}=\frac{\delta_{\mathcal{A}}}{2N(1-c)}$ \\ & $\beta_{\mathcal{B}}=\frac{\delta_{\mathcal{B}}}{2N(1-c)}$ \\ \hline
$I_{\mathcal{A}}=(1-c)N,I_{\mathcal{B}} = cN, c\in(0,1)$ & $\beta_{\mathcal{A}}=\frac{\delta_{\mathcal{A}}}{N}\frac{(1-c)}{c}$ \\ & $\beta_{\mathcal{B}}=\frac{\delta_{\mathcal{B}}}{N}\frac{c}{1-c}$ \\ \hline 
\textbf{Arbitrary Graph} &\\ \hline
$f=\hat{0}$ & $\rho(B\cdot A-\Delta)<0$ \\ \hline 
$f=$ perron eigenvector of & \\
$\Delta^{-1}BQA$ with eigenvalue $1$ & $\rho(B\cdot A-\Delta)>0$ \\ 
\hhline{|=|=|}\hhline{|=|=|}
$\forall i: f_i\in [a,b], 0<a<b<1$ & 
$\frac{a}{b(1-a)} < \frac{\beta_{\mathcal{A}}}{\delta_{\mathcal{A}}} d(i) < \frac{b}{a(1-b)}$ \\ &
$\frac{a}{b(1-a)} < \frac{\beta_{\mathcal{B}}}{\delta_{\mathcal{B}}} d(i) < \frac{b}{a(1-b)}$ 
\\ \hline
$\forall i \in \mathcal{B}: f_i\in [a,b], 0<a<b<1$ &  $\frac{\delta_{\mathcal{A}}}{\beta_{\mathcal{A}}}  \frac{a}{b(x-a)} < \frac{1}{x}d_{\mathcal{A}}(i)+$\\ &$+d_{\mathcal{B}}(i) < \frac{\delta_{\mathcal{A}}}{\beta_{\mathcal{A}}}  \frac{b}{a(x-b)}$

\\
$\forall i\in \mathcal{A}: f_i\in [\frac{a}{x},\frac{b}{x}], x>1$ & 
$\frac{\delta_{\mathcal{B}}}{\beta_{\mathcal{B}}}  \frac{a}{b(1-a)}<\frac{1}{x}d_{\mathcal{A}}(i)+$ \\& $+d_{\mathcal{B}}(i) < \frac{\delta_{\mathcal{B}}}{\beta_{\mathcal{B}}}  \frac{b}{a(1-b)}$
\\ \hline
\end{tabular}
\end{center}
\end{table}

Our basic assumption and innovation when compared to all previous approaches is that there is no fair-play and nodes have different profiles against the virus.  That is, the network is heterogeneous with respect to the virus, which means that nodes have different sensitivity to it. This is one of our main contributions in comparison with previous results where all nodes appear to have the same behavior towards the virus and the same model parameters \cite{Beutel:2012:IVN:2339530.2339601,Prakash:2012:WTC:2187836.2187975}. The propagation model which is followed, resembles the SIS (no immunity like flu) model where nodes are either susceptible or infected but with modifications. All nodes can get infected from one another, despite the difference of their profiles.


Our main contribution is that we provide answers for some of the questions above, for special topologies as the clique and arbitrary graphs of high or low connectivity. To the best of our knowledge, we are the first to provide theoretical and experimental findings on the propagation of a virus over a heterogeneous network.  We prove that in the case of two profiles, if one profile has high sensitivity to the virus and the other one has low sensitivity then nodes from both profiles will get infected in the case where the network is a clique. For arbitrary networks, we prove necessary conditions for the virus to die out allowing for multiple profiles (not just two). In the case of two profiles, we connect the degree of the nodes of the network to the footprint of the virus. The problem has many applications in the field of viral marketing, medicine, ecology, etc.. In Table~\ref{tab:res} we provide an overview of our theoretical findings. 

The outline of the paper is as follows: we provide a review of recent related work in Section \ref{related} and we give a detailed formulation of our model in Section \ref{problem_form}. Then, we present the analysis for every fixed point and the proof for our model in Section \ref{analysis}. Next, in section \ref{experiments} we verify our results with simulation experiments. Finally, we discuss issues and further extensions of our model and conclude in Sections~\ref{sec:discuss} and \ref{conclusion} respectively.

\section{Related Work} \label{related}
We provide an overview of related work on the field of epidemiology and information propagation. Known propagation models are the SIS model (no immunity), the SIR (lifetime immunity), SIRS (temporary immunity) etc. A survey for the available models in the field can be found in \cite{Hethcote:2000:MID:364897.364915}. One of the main topics in the field, is the epidemic threshold because of its role in the spread of an epidemic. It provides the conditions and appropriate parameters for a virus to reach the limit, which when crossed results in an epidemic. Earlier work in virus propagation models, include \textit{homogeneous} models where every member of the group has equal contact to others in the population \cite{McKendrick:1926,anderson:humaninf,bailey:infectious}. A lot of research focused on graphs of specific type such as random graphs, power-law graphs \cite{kephart:virus,pastor} and other.  However, the authors in \cite{ganesh} prove that for arbitrary topologies in the SIS model, the epidemic threshold depends on the leading eigenvalue of the adjacency matrix of the graph.

 The majority of the presented models have studied a single epidemic in a single topology whereas in later work \cite{Prakash:2012:WTC:2187836.2187975}, multiple virus models are introduced. All of them require that the network is fair-play, which means that all nodes have exactly the same behavior towards the viruses. In \cite{Prakash:2012:WTC:2187836.2187975}, an $SI_{1}I_{2}S$ model was used with two competing viruses that infect nodes of arbitrary topologies  where the nodes are mutually immune. The main result is that for any topology the stronger virus (above threshold) survives and wipes out the weaker one ("winner takes all"). In later work, the condition for mutual immunity is removed and the focus is on the conditions where nodes are infected from both viruses \cite{Beutel:2012:IVN:2339530.2339601}. In both cases, the model used is SIS-like whereas in \cite{Prakash:2011:TCA:2117684.2118312} the authors provide a generalized model that includes the majority of known epidemiological models with appropriate parameter definition in discrete time. All the aforementioned models, are applied to simple networks whereas recent work \cite{competingmemes}, refers to competing viruses in a composite network and conjectures that the stronger one will prevail over the weaker. Once again, the problem is formulated by a non-linear dynamical system.

Another section of related work is in the field of information diffusion. This includes apart from biological viruses, "viruses" with social, economic and market content, such as memes, competing products, rumors. As a result, research focus on information cascades, viral marketing, and competing product penetration \cite{domingos}. The models presented for information cascade, are divided in two classes, the Independent Cascade \cite{DBLP:conf:kdd:KempeKT03} and the Linear Threshold (LT) \cite{granovet}. An interesting approach is presented in \cite{kimura} where the authors construct a layered graph and apply bond percolation with a pruning strategy in order to efficiently estimate the influence function for the SIS model in social networks while in \cite{saito} the authors try to discover influential nodes in such networks. In the field of IC and LT, an important extension can be seen in \cite{barbieri}. Here the authors present a topic-aware extension of these well-know models and they propose a strategy for modelling authoritativeness, influence and relevance from such a perspective, a topic-aware perspective. 

Ecology, is another interesting application field for such propagation models. The principle of \textit{competitive exclusion} states that two species that compete for the exact same resources cannot stably coexist which resembles the aforementioned phenomenon of "winner takes all". This principle has been studied intensively using propagation models such as SIS, SIRS, Lotka-Volterra etc. but there has been no analytical solution presented so far \cite{competition},\cite{principle}, \cite{lotka} .

\section{Problem Formulation} \label{problem_form}

\begin{table}[tb]
\begin{center}
\caption{Frequently used symbols.}
\label{symbols}
\begin{tabular}{|p{0.16\columnwidth}||p{0.72\columnwidth}|}
\hline
{\bf Symbol} & {\bf Interpretation} \\ \hline\hline
$G=(V,E)$ & the graph $G$ with set of nodes $V$ and set of edges $E$ \\ \hline
$N$ & number of nodes in each profile or total number of nodes \\ \hline
$\mathcal{A}$, $\mathcal{B}$ & two sets of nodes that correspond to different profiles \\ \hline
$\beta_{\mathcal{A}}$ & infection rate for profile $\mathcal{A}$\\ \hline
$\delta_{\mathcal{A}}$ & healing rate for profile $\mathcal{A}$ \\ \hline
$p_{i,\mathcal{A}}$ & probability that node $i$ in profile $\mathcal{A}$ is infected \\ \hline
$p_{i}$ & probability that node $i$ is infected (profile is indifferent) \\ \hline
$d(i)$, $d_{\mathcal{A}}(i)$ & degree of node $i$ and degree of node $i$ w.r.t. profile $\mathcal{A}$ nodes only, respectively. \\ 
\hline
$I_{\mathcal{A}}$ & number of nodes infected in profile $\mathcal{A}$\\ \hline
$\mathcal{J}$ & the jacobian of the dynamical system\\ \hline
$A$ & adjacency matrix of the underlying graph\\ \hline
$\Delta$ & a diagonal matrix containing the healing rates ($\delta$) for each node (depends on the profile)\\ \hline
$B$ & a diagonal matrix containing the attack rates ($\beta$) for each node (depends on the profile)\\ \hline
$p$ & a state column vector with probabilities of infection ($p_i$) for each node \\ \hline
$Q$ & a diagonal matrix containing the probabilities ($1-p_i$) that nodes are not infected \\ \hline 
$\lambda$ & the eigenvalue of the corresponding matrix\\ \hline
$I$ & identity matrix of appropriate size\\ \hline
$\hat{x}$ & column vector full of $x$'s\\ \hline
\end{tabular}
\end{center}
\end{table}

We assume a SIS propagation model~\cite{Hethcote:2000:MID:364897.364915} that is applied on a heterogeneous network. That is, we assume that there is no fair game using the terminology of \cite{Prakash:2012:WTC:2187836.2187975,Beutel:2012:IVN:2339530.2339601}. 
Since this is the first theoretical treatment of heterogeneous environments for virus propagation we choose to work in the simple model of SIS. Even in this model, the analysis is quite complex. The heterogeneity of the network is realized by the existence of different profile agents with respect to the virus. That is, the virus has different behavior based on the profile of the agent. 

This is completely different when compared to the existence of multiple viruses over a network (e.g., \cite{Beutel:2012:IVN:2339530.2339601}). In the case of multiple viruses, the interesting issues are related to how the footprints of these viruses reaches an equilibrium subject to various rules of interaction. In our case, we assume a single virus which affects a network of multiple profiles. Until now there was no treatment (e.g., \cite{Prakash:2011:TCA:2117684.2118312}) of heterogeneous networks. To simplify the analysis we focus on a simple model (SIS) and we also assume the existence of two profiles, profile $\mathcal{A}$ and profile $\mathcal{B}$ - although in the general case we can allow for multiple profiles.

The following parameters define our problem:

\noindent \textbf{Healing Rate:} It is the death rate $\delta_{\mathcal{A}}$ of the virus in profile $\mathcal{A}$. This means that when an agent (a node in the network) is infected then the time taken to heal is exponentially distributed with respect to the parameter $\delta_{\mathcal{A}}$. Intuitively, a high $\delta$ value means that the time taken to heal is small while for low $\delta$ the time taken to heal is high and the virus persists. For example, imagine a rumor about PS3 in a social network where there are groups of PS4 fans (low $\delta$), groups of XBOX One fans (very high $\delta$) and other groups that have relatively high $\delta$.

\noindent \textbf{Attack Rate:} It is the infection rate of the virus. This depends on the endurance of the agent, which defines its profile with respect to the virus. An agent belonging to profile $\mathcal{A}$ has endurance $\beta_{\mathcal{A}}$ towards the virus, or alternatively, the virus has attack rate $\beta_{\mathcal{A}}$ towards agents in profile $\mathcal{A}$. In our game console setting, a high $\beta$ means that the agent is susceptible to the rumor about PS4 while a low $\beta$ means that the agent does not care about it since he may be a fan of a competitive product (XBOX One).

\subsection{Formal Problem Statement} 

We assume an undirected connected graph $G=(V,E)$, where the set of nodes $V$ corresponds to agents and the set of edges $E$ corresponds to an established relation through which the virus can infect other nodes. 

The general statement of the problem we wish to solve is the following: \emph{Given a network $G=(V,E)$ and the SIS parameters 
$\{\beta_{\mathcal{A}_1},\beta_{\mathcal{A}_2},\ldots\,\beta_{\mathcal{A}_k}\}$ and 
$\{\delta_{\mathcal{A}_1},\delta_{\mathcal{A}_2},\ldots\,\delta_{\mathcal{A}_k}\}$ 
for a set of profiles 
$\{\mathcal{A}_1,\mathcal{A}_2,\ldots\,\mathcal{A}_k\}$, 
determine the conditions under which the virus reaches a particular equilibrium state.}

In the following, most of the times we focus on the case of two profiles, in which case we name the second one $\mathcal{A}_2$ as $\mathcal{B}$.

\section{Analysis - Proofs} \label{analysis}

We first analyze the case where the graph $G$ is a clique of size $|V|=2N$. The reasons for this choice are twofold. First, we provide exact theoretical results. Second, we acquire intuition as to how the virus propagates on heterogeneous networks. Finally, we proceed with the analysis of arbitrary graphs.

\subsection{Proof Roadmap}
The proof consists of the following steps:

$1.$ \textbf{Dynamical System:} We construct a dynamical system of differential equations to approximate the virus propagation model and the process followed. In particular, the dynamical system is of the form
$ x'=F(x)$, where $x'$ is the component-wise derivative of $x$ and $F$: $\Re \rightarrow \Re$ is continuous and differentiable.

$2.$ \textbf{Fixed Points:}  We want to find the possible points where the dynamical system is in equilibrium and does not change state. Every point $\vec{x}$ ($\vec{x}$ is a vector) where $F(\vec{x})=0$ is considered to be such a point, called a \textit{fixed point}. There are several fixed points of the system, depending on the spread of the virus and the sensitivity/endurance level of the agents. However, we choose to present those that give better insight of the  system behavior.

$3.$ \textbf{Stability conditions:} We focus on the conditions required for each fixed point to be stable so that a possible perturbation will not push the system away from the equilibrium point. We are interested in the stability conditions of three scenarios:

$a.$ HIGH: Both profiles in the network have great endurance against the virus and tend to stay susceptible.

$b.$ LOW: Both profiles have equally low endurance against the virus and all agents get infected.

$c.$ MIXED: One profile has low endurance against the virus and the corresponding proportion of the  network tends to get infected while the other profile has high endurance and the corresponding agents tend to stay susceptible.
\\

In conclusion, after constructing a suitable dynamical system to describe the propagation process we continue with the analysis of the possible fixed points. Intuitively, one would require for each fixed point to be a stable attractor and not lead the system far away from the equilibrium point because of opposing forces. More specifically, we require each fixed point to be a hyperbolic fixed point, in which case the eigenvalues of the corresponding Jacobian matrix will not have a zero real part \cite{hirsch:differential}. With this requirement, a hyperbolic fixed point is stable when the eigenvalues have a negative real part and hence all the eigenvalues of the corresponding Jacobian matrix should satisfy this condition.

\subsection{Clique}

We start by analyzing the case of a clique with $2N$ nodes. We assume two profiles, $\mathcal{A}$ and $\mathcal{B}$ with $N$ nodes each.
Let $I_{\mathcal{A}}$ and $I_{\mathcal{B}}$ \footnote{These are time-dependent variables but for brevity we do not write them as such ($I_{\mathcal{A}}(t)$).} be the number of infected nodes. The following differential equations describe the evolution of the system:

\begin{equation}\label{eq:A}
\frac{dI_{\mathcal{A}}}{dt}=\beta_{\mathcal{A}}(N-I_{\mathcal{A}})(I_{\mathcal{A}}+I_{\mathcal{B}})-\delta_{\mathcal{A}} I_{\mathcal{A}}
\end{equation} 

\begin{equation}\label{eq:B}
\frac{dI_{\mathcal{B}}}{dt}=\beta_{\mathcal{B}}(N-I_{\mathcal{B}})(I_{\mathcal{B}}+I_{\mathcal{A}})-\delta_{\mathcal{B}} I_{\mathcal{B}}
\end{equation} 

Indeed, the change of $I_{\mathcal{A}}$ (similarly for $I_{\mathcal{B}}$) is equal to the new infected nodes of profile $\mathcal{A}$ due to the already infected nodes in both profiles (recall that because the graph is a clique all nodes connect to all other nodes) minus the already infected nodes of $I_{\mathcal{A}}$ that heal with probability $\delta_{\mathcal{A}}$.

These equations constitute a non-linear dynamical system and its Jacobian is:

\[
\mathcal{J}(I_{\mathcal{A}},I_{\mathcal{B}}) =
\left[ \begin{array}{cc}
\beta_{\mathcal{A}}(N-2I_{\mathcal{A}}-I_{\mathcal{B}})-\delta_{\mathcal{A}} 
& 
\beta_{\mathcal{A}}(N-I_{\mathcal{A}}) \\
\beta_{\mathcal{B}}(N-I_{\mathcal{B}}) 
& 
\beta_{\mathcal{B}}(N-2I_{\mathcal{B}}-I_{\mathcal{A}})-\delta_{\mathcal{B}}
\end{array} \right]
\]

We are only interested in the hyperbolic fixed points of this system since a hyperbolic equilibrium point is topologically equivalent to the orbit structure of the linearized dynamical system. We identify four interesting fixed points of the dynamical system:
\begin{enumerate}
 \item $I_{\mathcal{A}},I_{\mathcal{B}} \rightarrow 0$ where the virus dies out in both profiles.
 \item $I_{\mathcal{A}},I_{\mathcal{B}} \rightarrow N$ where the virus infects all nodes in both profiles.
 \item $I_{\mathcal{A}}\rightarrow 0, I_{\mathcal{B}} \rightarrow N$ where the virus infects one profile ($\mathcal{B}$ in this case) and leaves the other unaffected. There are two such fixed points that hold symmetrically for both profiles and they are the most interesting in this setting.
\end{enumerate}


For the discussion to follow, we assume w.l.o.g. that $\delta_{\mathcal{B}}>\delta_{\mathcal{A}}$ (the other case is symmetric in most cases).

\subsubsection{Fixed point: $I_{\mathcal{A}}\rightarrow 0,I_{\mathcal{B}}\rightarrow 0 $}
We compute the partial derivatives when the rates of change in $I_{\mathcal{A}}$ and $I_{\mathcal{B}}$ are zero, and the corresponding Jacobian for this fixed point will be:
\[
\mathcal{J}(0,0) =
\left[ \begin{array}{cc}
\beta_{\mathcal{A}}N-\delta_{\mathcal{A}} & \beta_{\mathcal{A}}N  \\
\beta_{\mathcal{B}}N & \beta_{\mathcal{B}}N-\delta_{\mathcal{B}} 
\end{array} \right]
\]

In order to compute the eigenvalues of the Jacobian matrix we solve: 
$det(\mathcal{J}(0,0)-\lambda I)=0 \Rightarrow \lambda^2+\lambda(\delta_{\mathcal{A}}+\delta_{\mathcal{B}}-N(\beta_{\mathcal{A}}+\beta_{\mathcal{B}}))+\delta_{\mathcal{A}}\delta_{\mathcal{B}}-N(\delta_{\mathcal{B}}\beta_{\mathcal{A}}+\delta_{\mathcal{A}}\beta_{\mathcal{B}})=0$ 

The discriminant of this quadratic equation is always positive (see Appendix~\ref{app:pos_discr} for the proof). Assuming that $\delta_{\mathcal{B}}>\delta_{\mathcal{A}}$ we get the following real eigenvalues:

\noindent $\lambda_{1,2}=\frac{N(\beta_{\mathcal{A}}+\beta_{\mathcal{B}})-(\delta_{\mathcal{A}}+\delta_{\mathcal{B}})\pm \sqrt{((\delta_{\mathcal{B}}-\delta_{\mathcal{A}})-N(\beta_{\mathcal{A}}+\beta_{\mathcal{B}}))^2+4N\beta_{\mathcal{A}}(\delta_{\mathcal{B}}-\delta_{\mathcal{A}})}}{2}$

This fixed point will be hyperbolic only when none of the eigenvalues of the corresponding Jacobian has a zero real part. In addition, the system will be stable at a hyperbolic fixed point only if the real part of the eigenvalues of the Jacobian is negative. Since the discriminant is $\Delta > 0$ the resulting eigenvalues will be real. For stability we require that $\lambda_{1},\lambda_{2}<0$. We find conditions only for the case where $\lambda_1<0$, since $\lambda_2<\lambda_1$. 

As for $\lambda_1$ we require that:
\noindent $N(\beta_{\mathcal{A}}+\beta_{\mathcal{B}})-(\delta_{\mathcal{A}}+\delta_{\mathcal{B}})<-\sqrt{((\delta_{\mathcal{B}}-\delta_{\mathcal{A}})-N(\beta_{\mathcal{A}}+\beta_{\mathcal{B}}))^2+4N\beta_{\mathcal{A}}(\delta_{\mathcal{B}}-\delta_{\mathcal{A}})}$
which gives:

\begin{equation}\label{cond:00}
N\frac{\delta_{\mathcal{A}}\beta_{\mathcal{B}}+\delta_{\mathcal{B}}\beta_{\mathcal{A}}}{\delta_{\mathcal{A}}\delta_{\mathcal{B}}}<1
\end{equation}

Just as a sanity check, imagine that there is only one profile, that is $\delta_{\mathcal{A}}=\delta_{\mathcal{B}}=\delta$ and $\beta_{\mathcal{A}}=\beta_{\mathcal{B}}=\beta$. In this case, we get the condition $\frac{2N\beta}{\delta}<1$ which is stated in~\cite{Prakash:2011:TCA:2117684.2118312}, since the largest eigenvalue of the adjacency matrix of a clique with $2N$ nodes is $2N-1$.

\subsubsection{Fixed point: $I_{\mathcal{A}}\rightarrow N,I_{\mathcal{B}}\rightarrow N$} \label{ssec:NN}

In fact, we assume that $I_{\mathcal{A}}\rightarrow cN,I_{\mathcal{B}}\rightarrow cN$, $0<c<1$, and $c \rightarrow 1$. Thus, we are going to substitute these values and find necessary conditions so that $c$ tends to $1$. Substituting these values in Equations~(\ref{eq:A}) and (\ref{eq:B}) when $\frac{dI_{\mathcal{A}}}{dt}=\frac{dI_{\mathcal{B}}}{dt}=0$ we get:
\begin{equation}\label{eq:cna}
\beta_{\mathcal{A}}=\frac{\delta_{\mathcal{A}}}{2N(1-c)}
\end{equation}
\begin{equation}\label{eq:cnb}
\beta_{\mathcal{B}}=\frac{\delta_{\mathcal{B}}}{2N(1-c)}
\end{equation}

Increasing the values of $\beta_{A}$ or $\beta_{B}$, the value of $c$ will increase accordingly for the same healing rates and network size. Having proved that these are fixed points (by construction) we move to proving that they are hyperbolic and stable and get the necessary conditions. We get:
\[
\mathcal{J}(cN,cN) =
\left[ \begin{array}{cc}
(1-3c)\beta_{\mathcal{A}}N-\delta_{\mathcal{A}} 
& 
(1-c)\beta_{\mathcal{A}}N \\
(1-c)\beta_{\mathcal{B}}N
& 
(1-3c)\beta_{\mathcal{B}}N-\delta_{\mathcal{B}}
\end{array} \right]
\]

To find the eigenvalues of this jacobian we need to solve the following equation: $\lambda^2+\lambda(\delta_{\mathcal{A}}+\delta_{\mathcal{B}}+(3c-1)N(\beta_{\mathcal{A}}+\beta_{\mathcal{B}}))+\delta_{\mathcal{A}}\delta_{\mathcal{B}}+4c(2c-1)\beta_{\mathcal{A}}\beta_{\mathcal{B}}N^2+(3c-1)N(\beta_{\mathcal{A}}\delta_{\mathcal{B}}+\beta_{\mathcal{B}}\delta_{\mathcal{A}})=0$.

The discriminant is:
\noindent $\Delta=(\delta_{\mathcal{A}}+\delta_{\mathcal{B}}+(3c-1)N(\beta_{\mathcal{A}}+\beta_{\mathcal{B}}))^2-4\delta_{\mathcal{A}}\delta_{\mathcal{B}}-4(3c-1)N(\delta_{\mathcal{B}}\beta_{\mathcal{A}}+\delta_{\mathcal{A}}\beta_{\mathcal{B}})-16c(2c-1)\beta_{\mathcal{A}}\beta_{\mathcal{B}}N^2$

\noindent which results in:

\noindent 
$\Delta=(\delta_{\mathcal{A}}+\delta_{\mathcal{B}}+(3c-1)N(\beta_{\mathcal{A}}+\beta_{\mathcal{B}}))^2-4\delta_{\mathcal{A}}\delta_{\mathcal{B}}-12cN\beta_{\mathcal{A}}\delta_{\mathcal{B}}-12cN\beta_{\mathcal{B}}\delta_{\mathcal{A}}+4N\beta_{\mathcal{A}}\delta_{\mathcal{B}}+4N\beta_{\mathcal{B}}\delta_{\mathcal{A}}-32c^2\beta_{\mathcal{A}}\beta_{\mathcal{B}}N^2+16c\beta_{\mathcal{A}}\beta_{\mathcal{B}}N^2$

By using Equations~\ref{eq:cna} and \ref{eq:cnb} and replacing above we get:

\noindent $\Delta=\left(\delta_{\mathcal{A}}+\delta_{\mathcal{B}}+(3c-1)\left(\frac{\delta_{\mathcal{A}}+\delta_{\mathcal{B}}}{2(1-c)}\right)\right)^2-4\delta_{\mathcal{A}}\delta_{\mathcal{B}}-12c\frac{\delta_{\mathcal{A}}\delta_{\mathcal{B}}}{(1-c)}+4\frac{\delta_{\mathcal{A}}\delta_{\mathcal{B}}}{(1-c)}-8c^2\frac{\delta_{\mathcal{A}}\delta_{\mathcal{B}}}{(1-c)^2}+4c\frac{\delta_{\mathcal{A}}\delta_{\mathcal{B}}}{(1-c)^2}\Rightarrow$
\[
\Delta=\frac{1}{4(1-c)^2}\left[(\delta_{\mathcal{A}}+\delta_{\mathcal{B}})^2(1+c)^2-4c\delta_{\mathcal{A}}\delta_{\mathcal{B}}\right]
\]
This means that the eigenvalues are the following:

\[
\lambda_{1,2}=\frac{-(c+1)(\delta_{\mathcal{A}}+\delta_{\mathcal{B}})\pm\sqrt{(c+1)^2(\delta_{\mathcal{A}}+\delta_{\mathcal{B}})^2-4c\delta_{\mathcal{A}}\delta_{\mathcal{B}}}}{4(1-c)}
\]

Let $x=(c+1)(\delta_{\mathcal{A}}+\delta_{\mathcal{B}})$ and $y=4c\delta_{\mathcal{A}}\delta_{\mathcal{B}}$. Then, by imposing $\frac{1}{2}<c<1$, it holds that $x,y>0$ and it follows that the real part of $\lambda_1=\frac{-x-\sqrt{x^2-y}}{4(1-c)}$ is always negative since $x^2-y<x^2$ and if $x^2-y$ is negative then its square root will be of the form $zi$, where $z\in \mathbb{R}$ and $i=\sqrt{-1}$. For the same reasons, the real part of $\lambda_2=\frac{-x+\sqrt{x^2-y}}{4(1-c)}$ is always negative. In fact, one can prove that $\Delta=x^2-y>0$ but it requires a lot of algebraic handling.

The values for $\beta_{\mathcal{A}}$ and $\beta_{\mathcal{B}}$ as determined by Equations~\ref{eq:cna} and \ref{eq:cnb}, for $c\in \left(\frac{1}{2},1\right)$ ensure that the fixed point is at $cN$ for both profiles. Note that in contrast to the fixed point where the virus dies out here we do not get a condition in the form of an inequality since this is much harder to get. However, we get a relationship between all parameters so that the virus infects all nodes when $c \rightarrow 1$. In addition, we get necessary conditions so that a part of the network is infected (e.g., imagine choosing $c=2/3$). Finally, $\beta_{\mathcal{A}}$ (the same holds for $\beta_{\mathcal{B}}$) should be less than $1$ which means that $c<\frac{2N-\delta_{\mathcal{A}}}{2N}$ and $c<\frac{2N-\delta_{\mathcal{B}}}{2N}$.

\subsubsection{Fixed point $I_{\mathcal{A}}\rightarrow 0,I_{\mathcal{B}}\rightarrow N $}

In fact we are going to prove necessary conditions for the case where $I_{\mathcal{A}}\rightarrow (1-c)N,I_{\mathcal{B}}\rightarrow cN$, where $0<c<1$. Note, that having one profile completely infected with the virus allows for some infections in the other profile, which may be resilient, since there will be some nodes infected with positive probability.

We substitute these values in Equations~\ref{eq:A} and \ref{eq:B} and we get a relation between the parameters:
\begin{equation}\label{cond:1cN}
\beta_{\mathcal{A}}=\frac{\delta_{\mathcal{A}}}{N}\frac{1-c}{c} \mbox{, }
\beta_{\mathcal{B}}=\frac{\delta_{\mathcal{B}}}{N}\frac{c}{1-c}
\end{equation}

Thus, when these two relations hold then we guarantee that $(1-c)N$ and $cN$ are fixed points. In the following, we assume w.l.o.g. that $c\rightarrow 1$ and thus profile $\mathcal{A}$ is the resilient one while profile $\mathcal{B}$ is the more susceptible one.  

Our main work is to examine the necessary stability conditions. The Jacobian matrix in this case is the following:

\[
\mathcal{J}((1-c)N,cN) =
\left[ \begin{array}{cc}
(c-1)\beta_{\mathcal{A}}N-\delta_{\mathcal{A}} 
& 
c\beta_{\mathcal{A}}N \\
(1-c)\beta_{\mathcal{B}}N 
& 
-c\beta_{\mathcal{B}}N-\delta_{\mathcal{B}}
\end{array} \right]
\]

In order to compute the eigenvalues of the Jacobian matrix we solve the following quadratic equation: $\lambda^2+\lambda(\delta_{\mathcal{A}}+\delta_{\mathcal{B}}+c\beta_{\mathcal{B}}N+(1-c)\beta_{\mathcal{A}}N)+\delta_{\mathcal{A}}\delta_{\mathcal{B}}+(1-c)\beta_{\mathcal{A}}\delta_{\mathcal{B}}N+c\delta_{\mathcal{A}}\beta_{\mathcal{B}}N=0$.

Similarly to~\ref{ssec:NN}, let $x=\frac{c^2-c+1}{c(1-c)}(c\delta_{\mathcal{B}}+(1-c)\delta_{\mathcal{A}})$ and $y=4\delta_{\mathcal{A}}\delta_{\mathcal{B}}\frac{2c^2-2c+1}{c(1-c)}$. Note that $x,y>0$ for $c\in (0,1)$ and the characteristic polynomial can be written as $\lambda^2+x\lambda +\frac{1}{4}y=0$. Then, the eigenvalues are $\lambda_{1,2}=\frac{-x\pm\sqrt{x^2-y}}{2}$.

The real part of $\lambda_1=\frac{-x-\sqrt{x^2-y}}{2}$ is always negative since $x^2-y<x^2$ and if $x^2-y$ is negative then its square root will be of the form $zi$, where $z\in \mathbb{R}$ and $i=\sqrt{-1}$. For the same reasons, the real part of $\lambda_2=\frac{-x+\sqrt{x^2-y}}{2}$ is always negative.

Similarly to~\ref{ssec:NN}, we provide two equalities that define the relationship between the parameters of the problem. We can choose $c$, $\delta_{\mathcal{A}}$ and $\delta_{\mathcal{B}}$, compute $\beta_{\mathcal{A}}$ and $\beta_{\mathcal{B}}$ and then the resulting dynamical system tends to have a footprint of $(1-c)N$ for profile $\mathcal{A}$ and $cN$ for profile $\mathcal{B}$.

\subsection{Arbitrary Simple Relation Graphs}\label{arbitrary2}

Let $A$ be the adjacency matrix of the arbitrary connected graph $G$ with $N$ nodes. Let $p_{i,\mathcal{A}_k}$ be the probability of node $i$ in profile $\mathcal{A}_k$ to be infected. Finally, let $p_i$ be the probability that node $i$ is infected when the profile of $i$ is irrelevant.

In this case the dynamical system that describes the evolution for all profiles $\mathcal{A}_k$ is the following: 

\begin{equation}
\frac{dp_{i,\mathcal{A}_k}}{dt}=-\delta_{\mathcal{A}_k}p_{i,\mathcal{A}_k}+\beta_{\mathcal{A}_k}(1-p_{i,\mathcal{A}_k})\sum\limits_{j} (1_{j} A_{ji})
\end{equation}
where $1_{j}$ is the indicator random variable denoting whether node $j$ is infected from the virus. Due to the presence of random variables $1_{j}$, our system is not a Markov chain. By applying a first order mean-field approximation (see~\cite{gould2010statistical} for a nice presentation) we deliberately assume that these indicator variables are equal to their expected value. As such, $1_{j}\approx E[1_{j}]=p_{j}$, and thus we get the following equations for each node $i$:

\begin{equation}
\frac{dp_{i,\mathcal{A}_k}}{dt}=-\delta_{\mathcal{A}_k}p_{i,\mathcal{A}_k}+\beta_{\mathcal{A}_k}(1-p_{i,\mathcal{A}_k})\sum\limits_{j} (p_{j} A_{ji})
\end{equation}

We are interested in computing necessary conditions for fixed points. At a fixed point it holds that $\frac{dp_{i,\mathcal{A}_k}}{dt}=0$ and thus we get:

\begin{equation} \label{eq:fixeda}
\delta_{\mathcal{A}_k}p_{i,\mathcal{A}_k}=\beta_{\mathcal{A}_k}(1-p_{i,\mathcal{A}_k})\sum\limits_{j} (p_{j} A_{ji})
\end{equation}

\noindent Writing this in a vector form we get:

\begin{equation}\label{eq:arb_fixed}
p = \Delta^{-1} B Q A p
\end{equation}
where $p$ is the state column vector $[p_1,p_2,\ldots,p_N]^T$. In addition, $\Delta$ is a diagonal $N\times N$ matrix, where element $(i,i)$ is equal to $\delta_{\mathcal{A}}$ if node $i$ belongs to profile $\mathcal{A}$ and $\delta_{\mathcal{B}}$ if it belongs to profile $\mathcal{B}$. Similarly, $B$ is a $N\times N$ diagonal matrix containing values $\beta_{\mathcal{A}}$ or $\beta_{\mathcal{B}}$ depending on the profile of the corresponding node. Finally, $Q$ is also a diagonal $N\times N$ matrix where its element $(i,i)$ contains the probability $1-p_i$, that is $Q=I-diag(p)$ and represents the probability that a node is not infected. Note also that the adjacency matrix $A$ is symmetric and as a result $A=A^T$. 

This is a homogeneous non-linear system of equations, which makes finding the solutions a rather herculean task. This is why, apart from the obvious zero solution, all other interesting fixed points and stability conditions will be only qualitatively characterized.

\subsubsection{The Zero Fixed Point and Stability Condition} \label{sssec:zero}

The jacobian of Equation~\ref{eq:arb_fixed} can be written as:
\begin{equation}\label{eq:jacobian}
\mathcal{J}(p)=-I+\Delta^{-1}BQA-\Delta^{-1}B diag(Ap)
\end{equation}
where $diag(Ap)$ is the diagonal $N\times N$ matrix, whose diagonal contains the elements of the column vector $Ap$. 

Apparently, $p=\hat{0}$, where $\hat{0}$ is the column vector full of zeros, is a fixed point since it is a solution of Equation~\ref{eq:arb_fixed}. The fixed point $p=\hat{0}$ is stable if it holds that the real part of all eigenvalues of $\mathcal{J}(\hat{0})=\Delta^{-1}BA-I$ ($Q=I$ in this case) are negative. As a sanity check, if $\beta=\beta_{\mathcal{A}}=\beta_{\mathcal{B}}$ and $\delta=\delta_{\mathcal{A}}=\delta_{\mathcal{B}}$, which means that there is only one profile, then if $\lambda$ is the largest eigenvalue of $A$ the largest eigenvalue of $\mathcal{J}(\hat{0})$ will be $\frac{\beta\lambda}{\delta}-1$. Thus, the fixed point will be stable if $\frac{\beta\lambda}{\delta}-1 <0 \Rightarrow \frac{\beta\lambda}{\delta}<1$, which in fact is the result provided in \cite{Prakash:2011:TCA:2117684.2118312} for one profile and one virus for the SIS model.

\subsubsection{Other Fixed Points}

What happens when the healing and infection rates are such so that $0$ is not a stable fixed point? This happens when $\rho(B\cdot A-\Delta)>0$. We start by looking at point $p=\hat{1}$, where $\hat{1}$ is the column vector full of ones. Then, the right hand-side of Equation~\ref{eq:arb_fixed} becomes zero (since $Q=0$) and the only case for the equality to hold is if $\Delta=0$ (of course all infection rates should be non-zero). This means that all nodes will definitely be in an infected state if the probability to heal is zero, which is a sound conclusion.

We start with some easy facts about the matrix $\Delta^{-1} BQA$.

\begin{lemma} \label{lem:irreducible}
Matrix $\Delta^{-1} BQA$ is non-negative and irreducible. 
\end{lemma}
\begin{proof}
Matrix $\Delta^{-1} BQ$ is a diagonal positive matrix, assuming that all infection and healing rates are non-zero. As such, matrix $\Delta^{-1} BQA$ is non-negative since $\Delta^{-1} BQ$ is non-negative and $A$ is also non-negative. Finally, $A$ is irreducible since it is the adjacency matrix of a connected simple graph whose entries are simply multiplied by the diagonal entries of $\Delta^{-1} BQ$.
\end{proof}

\begin{lemma}
$\Delta^{-1} BQA$ has a positive real eigenvalue $\lambda$ as its largest in absolute values. The multiplicity of $\lambda$ is $1$ and it has the only positive eigenvector. 
\end{lemma}
\begin{proof}
By Perron-Frobenius theorem $\lambda$ is called Perron-Frobenius (PF) eigenvalue with its corresponding PF eigenvector.
\end{proof}

\begin{lemma} \label{lem:fixed}
The PF eigenvector $f$ is a fixed point.
\end{lemma}
\begin{proof}
\[
\mathcal{J}(f)=-I+\Delta^{-1}B(I-diag(f))A-\Delta^{-1}B diag(Af) 
\]
All eigenvalues of $\mathcal{J}(f)$ are negative. 
\[
\lambda(\mathcal{J}(f))\leq -1 + \lambda\left(\Delta^{-1}B(I-diag(f))A\right) +\lambda\left(-\Delta^{-1}B diag(Af)\right)
\] 
However, we know that $\lambda\left(\Delta^{-1}B(I-diag(f))A\right) \leq 1$ and thus we get:
\[
\lambda(\mathcal{J}(f))\leq \lambda\left(-\Delta^{-1}B diag(Af)\right)
\]
Since $\Delta^{-1}B diag(Af)$ is a positive diagonal matrix we get that $\lambda\left(-\Delta^{-1}B diag(Af)\right) \leq 0$ which proves the fact that $\lambda(\mathcal{J}(f))\leq 0$.
\end{proof}

This PF eigenvector does not really provide us with enough information. In fact, it seems pretty hard and rather overly optimistic to get a clean result. This is why focusing back in the case of two profiles, we are interested in fixed points that correspond to the following two cases: a) almost all nodes are infected in both profiles and b) almost all nodes of one profile are infected and the nodes in the other profile are healthy. We will focus on specific cases of graphs and provide a rather qualitative explanation of results.

\subsubsection{The Case of Two Profiles}

Assume two profiles $\mathcal{A}$ and $\mathcal{B}$. We also assume that the edges of the network have no weights and as a result the adjacency matrix is a 0/1 matrix. Then the PF eigenvector $f$ could be written as 
$f=
\left[ \begin{array}{c}
f_{\mathcal{A}} \\
f_{\mathcal{B}}
\end{array} \right]$
corresponding to the two profiles. We write differently Equation~\ref{eq:fixeda} for the two profiles as follows:

\begin{equation} \label{eq:fixeda-new}
\sum\limits_{j\in \mathcal{B}} (p_{j} A_{ji}) + \sum\limits_{j\in \mathcal{A}} (p_{j} A_{ji})= \frac{\delta_{\mathcal{A}}}{\beta_{\mathcal{A}}} \frac{p_{i,\mathcal{A}}}{1-p_{i,\mathcal{A}}} 
\end{equation}

\begin{equation}\label{eq:fixedb-new}
\sum\limits_{j\in \mathcal{A}} (p_{j} A_{ji}) + \sum\limits_{j\in \mathcal{B}} (p_{j} A_{ji})= \frac{\delta_{\mathcal{B}}}{\beta_{\mathcal{B}}} \frac{p_{i,\mathcal{B}}}{1-p_{i,\mathcal{B}}} 
\end{equation}

We first look at the case where both profiles are susceptible to the virus and the network is flooded. As previously mentioned, we only provide a qualitative analysis by going backwards in some sense. In particular, we assume that $\forall i$, $p_i\in [a,b]$, for $0 <a < b <1$. Then, we get the following inequalities by Equations~\ref{eq:fixeda-new} and \ref{eq:fixedb-new}.
\[\forall i: \frac{\beta_{\mathcal{A}}}{\delta_{\mathcal{A}}} d(i) < \frac{b}{a(1-b)}
 \mbox{, }
\forall i: \frac{\beta_{\mathcal{A}}}{\delta_{\mathcal{A}}} d(i) > \frac{a}{b(1-a)}\]
\[\forall i: \frac{\beta_{\mathcal{B}}}{\delta_{\mathcal{B}}} d(i) < \frac{b}{a(1-b)}
 \mbox{, }
\forall i: \frac{\beta_{\mathcal{B}}}{\delta_{\mathcal{B}}} d(i) > \frac{a}{b(1-a)}\]
where $d(i)$ is the degree of node $i$. Apparently this describes a rather limited (albeit infinite) family of graphs. 

For example, assume that $\frac{\beta_{\mathcal{A}}}{\delta_{\mathcal{A}}}=2$ and $\frac{\beta_{\mathcal{B}}}{\delta_{\mathcal{B}}}=4$ and let $a=0.6$ and $b=0.9$. Then, we get that the following restrictions should hold on this graph to show indeed such behavior with respect to the virus. 
\[d(i)< 7.5, d(i) > 0.8, d(i) < 3.75, d(i) > 0.4\]
This means that all graphs with minimum degree $1$ and maximum degree $3$ will have a PF eigenvector whose entries will be in the range $[0.6, 0.9]$ for the particular choice of parameters. 

Now we move to the case where $f_{\mathcal{A}}$ has small entries and $f_{\mathcal{B}}$ has large ones. This means that $\mathcal{B}$ is susceptible to the virus while $\mathcal{A}$ is not. Thus, we assume that all entries of $f_{\mathcal{B}}$ are in the range $[a,b]$ and all entries in $f_{\mathcal{A}}$ in the range $\left[\frac{a}{x},\frac{b}{x}\right]$. Parameter $x>1$ allows for simplifications and expresses how smaller the probability for nodes in $\mathcal{B}$ is w.r.t. the probability of nodes in $\mathcal{A}$.

Similarly, we get the following inequalities by Equations~\ref{eq:fixeda-new} and \ref{eq:fixedb-new}.
\[\forall i: \frac{\delta_{\mathcal{A}}}{\beta_{\mathcal{A}}}  \frac{a}{b(x-a)} < \frac{1}{x}d_{\mathcal{A}}(i)+d_{\mathcal{B}}(i) < \frac{\delta_{\mathcal{A}}}{\beta_{\mathcal{A}}}  \frac{b}{a(x-b)}
\]
\[\forall i: \frac{\delta_{\mathcal{B}}}{\beta_{\mathcal{B}}}  \frac{a}{b(1-a)} < \frac{1}{x}d_{\mathcal{A}}(i)+d_{\mathcal{B}}(i) < \frac{\delta_{\mathcal{B}}}{\beta_{\mathcal{B}}}  \frac{b}{a(1-b)}
\]
where $d_{\mathcal{A}}(i)$ is the degree of node $i$ w.r.t. profile $\mathcal{A}$ and similarly is defined $d_{\mathcal{B}}(i)$. 

For example, assume that $\frac{\delta_{\mathcal{A}}}{\beta_{\mathcal{A}}}=10^3$ and $\frac{\delta_{\mathcal{B}}}{\beta_{\mathcal{B}}}=0.01$ and let $a=0.8$, $b=0.99$ and $x=10^2$. Then, we get that the following restrictions should hold on this graph to show indeed such behavior with respect to the virus. 
\[\frac{1}{x}d_{\mathcal{A}}(i)+d_{\mathcal{B}}(i) < 12.4, \frac{1}{x}d_{\mathcal{A}}(i)+d_{\mathcal{B}}(i) > 8.1 \]
\[\frac{1}{x}d_{\mathcal{A}}(i)+d_{\mathcal{B}}(i) < 61.9 , \frac{1}{x}d_{\mathcal{A}}(i)+d_{\mathcal{B}}(i) > 2.02\]
This means that all graphs with minimum degree $9$ and maximum degree $12$ will have a PF eigenvector whose entries will be in the range $[0.8, 0.99]$ for $f_{\mathcal{B}}$ and in the range $[0.008,0.0099]$ for $f_{\mathcal{A}}$ for the particular choice of all the other parameters.

\section{Experimental Results} \label{experiments}

We evaluate our results using simulation experiments on various synthetic datasets including publicly available datasets. 
In particular, the datasets used are:
\begin{itemize}
\item Clique: A clique graph of $2000$ nodes.
\item Arbitrary topologies: Arbitrary graphs with varying average node degree, including powerlaw graphs.
\item Enron email network: The Enron email communication network \cite{emailenron} covers all the email communication within a dataset of email addresses. Nodes of the network are email addresses and if an address $i$ sent at least one email to address $j$, the graph contains an undirected edge from $i$ to $j$. The graph consists of $36692$ nodes and $183831$ edges. For all non-Enron email addresses we can only observe their communication with Enron email addresses.
\item Montgomery: A physical contact graph, representing the synthetic population of Montgomery County \cite{montgomery}, which contains $77,820$ people interacting with each other during their daily activities. The total number of activities for the population is $429,590$, which are conducted in $26,941$ distinct locations (besides home locations). The resulting social contact network has $77,820$ nodes (one per person) and $2,019,220$ edges.
\end{itemize}

In our experiments we implemented a discrete-time simulation in Java of the SIS model with two profiles $\mathcal{A}$ and $\mathcal{B}$, unless otherwise stated. In the experimental evaluation the profiles are set by dividing randomly the respective dataset in two equal parts based on the size of the graph. In one case the profiles are specified based on age groups and we examine the results of such profiling. For the experiments regarding the clique graph, the initial infection consists of $10$ nodes in each profile and every simulation is run for $2000$ time steps. In the contact networks (Enron email and Montgomery), the initial infection consists of the 5$\%$ of nodes with maximum degree (in range) and every simulation lasts for $5000$ rounds. A different approach is followed in powerlaw graphs. The graphs are generated using Boost library (C++), they have $5000$ nodes and every simulation lasts for $5000$ rounds. However, here we infect an amount of 5\% of nodes which have low degree. This is due to the fact that these graphs that follow a powerlaw distribution, they have many nodes with low degree while a few nodes with high degree. Consequently, we choose to infect an amount of low-degree nodes in order to observe how the virus propagates in the network since this is not possible with high-degree nodes.

We run experiments that verify our results for each fixed point for the following cases: a) Both profiles in the network have great endurance against the virus and tend to stay unaffected by the virus, b) both profiles have equally low endurance against the virus and all agents get infected and c) one profile has low endurance against the virus and the corresponding proportion of the  network tends to get infected while the other profile has high endurance and the corresponding agents tend to stay susceptible. 

\subsection{Simulation Results}
\subsubsection{Simulations for clique}
Figure~\ref{clique} demonstrates our results for the clique for various cases. In Figure \ref{clique}(a), we observe the case where the nodes have very high endurance against the virus and as a result the infected nodes in both profiles heal and the virus dies out. The figure depicts the number of infected nodes of both profiles and the total amount of infected nodes in the clique versus time. The used parameters that satisfy condition \ref{cond:00} are $(\beta_{\mathcal{A}},\delta_{\mathcal{A}} )=(0.0000005,0.01)$ and $(\beta_{\mathcal{B}},\delta_{\mathcal{B}} ) =(0.0000009,0.01)$. While the simulation was run for $2000$ rounds the fixed point $(0,0)$ is stable and as a result the system converges to this point very quickly, in $200$ rounds. As a result we omitted to depict the amount of infected nodes for rounds ($201-5000$) since it is still down to zero.

In Figure \ref{clique}(b) we assume a network where the nodes from both profiles $\mathcal{A}$ and $\mathcal{B}$ have the following healing and  infection rates $(\beta_{\mathcal{A}},\delta_{\mathcal{A}} )=(0.01,0.0005)$ and $(\beta_{\mathcal{B}},\delta_{\mathcal{B}} ) =(0.03,0.0006)$ which satisfy conditions \ref{eq:cna} and \ref{eq:cnb}. Here, the stable point is $(cN, cN)$ where $c=0.99$. It is expected that the majority of nodes from both profiles, will get infected and the equilibrium point will be reached when almost $cN$ nodes get infected. We observe that the system converges very quickly due to high infection rate and consequently this is a stable fixed point since for all rounds of execution, the amount of infected nodes from both profiles is steadily up to 2000. The amount of infected nodes in profile $\mathcal{A}$ in Figure \ref{clique}(b), is covered from the amount of infected nodes in profile $\mathcal{B}$ since they are equal. We forgot to mention that experimentally, we consider a fixed point as stable when the amount of infected nodes for the majority of rounds, is almost the same.
The value of variable $c$ affects the amount of nodes that get infected as well as the time required for the system convergence. Instinctively, we expect that for lower values, the amount of infected nodes will be smaller and the time required for the system to reach an equilibrium state will be increased. 

Finally, in the case where profile $\mathcal{A}$ has low sensitivity against the virus while profile $\mathcal{B}$ has high sensitivity, the stable fixed point according to our results, should be $((1-c)N,cN)$. As it can be seen in Figure \ref{clique}(c), the amount of infected nodes in profile $\mathcal{A}$ is very low in comparison with the amount of infected nodes in profile $\mathcal{B}$. In this case, we have assumed that $c=0.99$ where infection rates and healing rates, satisfying the conditions in Equation~\ref{cond:1cN}, are $(\beta_{\mathcal{A}},\delta_{\mathcal{A}} )=(0.0000055,0.1)$ and $(\beta_{\mathcal{B}},\delta_{\mathcal{B}} )=(0.1,0.01)$. As it can be seen, this is also a stable fixed point since for all rounds of execution, the amount of infected nodes from profile $\mathcal{B}$ is steadily up to 1000 while for profile $\mathcal{A}$ is down to zero.The amount of infected nodes in profile $\mathcal{B}$ in Figure \ref{clique}(c), is covered from the amount of totally infected nodes.

%

\begin{figure*}[!tb]
\hfill
\begin{minipage}{0.32\textwidth}
  \epsfig{file=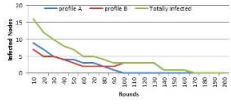,scale=0.4}
  \label{fig:clique00}
  \subcaption{$(0,0)$}
\end{minipage}
\hfill
\begin{minipage}{0.32\textwidth}
 \epsfig{file=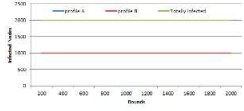,scale=0.4}
   \label{fig:cliquecN2}
   \subcaption{$(cN,cN)$}
\end{minipage}
\hfill
\begin{minipage}{0.32\textwidth}%
 \epsfig{file=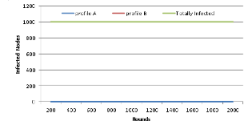,scale=0.4}
  \label{fig:clique1cN}
  \subcaption{$((1-c)N,cN)$}
\end{minipage}
\hfill
\caption{Simulations in a clique graph of $2000$ nodes. Figure (a) depicts the case where both profiles have low sensitivity against the virus. Figure (b) depicts the case where both profiles have high sensitivity against the virus and c=0.9. In Figure (c), profile $\mathcal{A}$ has low sensitivity against the virus while profile $\mathcal{B}$ has high sensitivity and c=0.9 }
\label{clique}
\end{figure*}

\subsubsection{Simulations for arbitrary graphs}
In the case of arbitrary graphs, the experiments were executed using the Enron email network that was described in the beginning of the experimental section. Using this network, we verify our theoretical results for the zero fixed point and we evaluate the existence of other fixed points in the case of two profiles. Similarly to the clique experiments, the dataset is divided to two profiles based on the size of the graph (half nodes in each profile) while the simulations last for $5000$ rounds. Here, as in all social contact graphs, we initially infect the top 5$\%$ of nodes according to their degree. It has been noticed that social graphs tend to create small highly connected subgraphs while weak connections exist between them. In order to effectively initialize the simulation, we chose to infect nodes that have relatively high degree.

To verify the zero fixed point, we used the following infection and healing parameter values:$(\beta_{\mathcal{A}},\delta_{\mathcal{A}})=(0.0009,0.5)$ and $(\beta_{\mathcal{B}},\delta_{\mathcal{B}})=(0.0005,0.7)$. In Figure \ref{emailenron_00}, we can observe that the system converges to zero in the first $12$ rounds, presenting the same behavior for the rest of the total $5000$ rounds. (We omit the rest of rounds in order to present the rapid convergence in the first rounds while for the rest part of the simulation, the result as expected is stable at zero).
 
\begin{figure}
\begin{center}
\epsfig{file=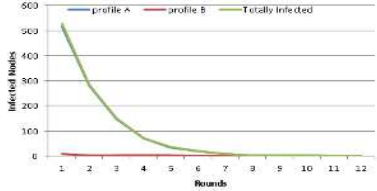,height=1in, width=2in}
\end{center}
\caption{Simulation where both profiles have low sensitivity against the virus in the Enron email network.}
 \label{emailenron_00}
\end{figure}

Next, we examine the case where both profiles have very high sensitivity against the virus and as a result, the network is flooded due to virus propagation. In this case, the infection and healing parameter values that were used, are $\beta_{\mathcal{A}}=\beta_{\mathcal{B}}=0.006$ and $\delta_{\mathcal{A}}=\delta_{\mathcal{B}}=0.0001$ for $ (\alpha,b)=(0.001,0.99) $. These parameter values satisfy the general conditions in Equations~\ref{eq:fixeda-new}, \ref{eq:fixedb-new} according to which, the node degree should be in the range $(1-1650)$. From the initial infection process, $1866$ nodes get infected and the initial status as well as the final status of the graph is depicted in Figure \ref{emailenron_full}.

\begin{figure*}[!tb]
\hfill
\begin{minipage}{0.32\textwidth}
  \epsfig{file=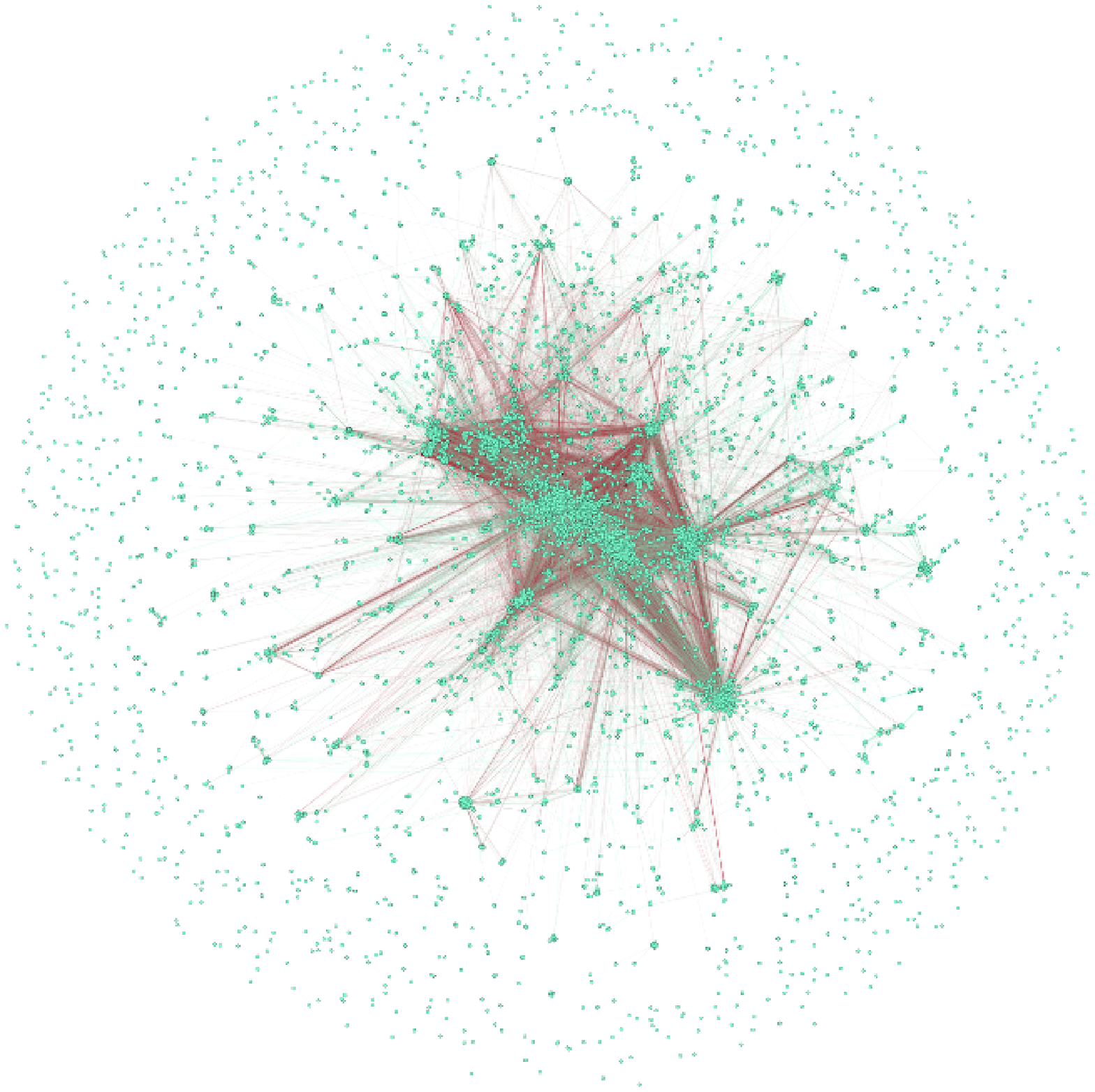,scale=0.3}
  \label{fig:init_infected}
  \subcaption{Initial Status}
\end{minipage}
\hfill
\begin{minipage}{0.32\textwidth}
 \epsfig{file=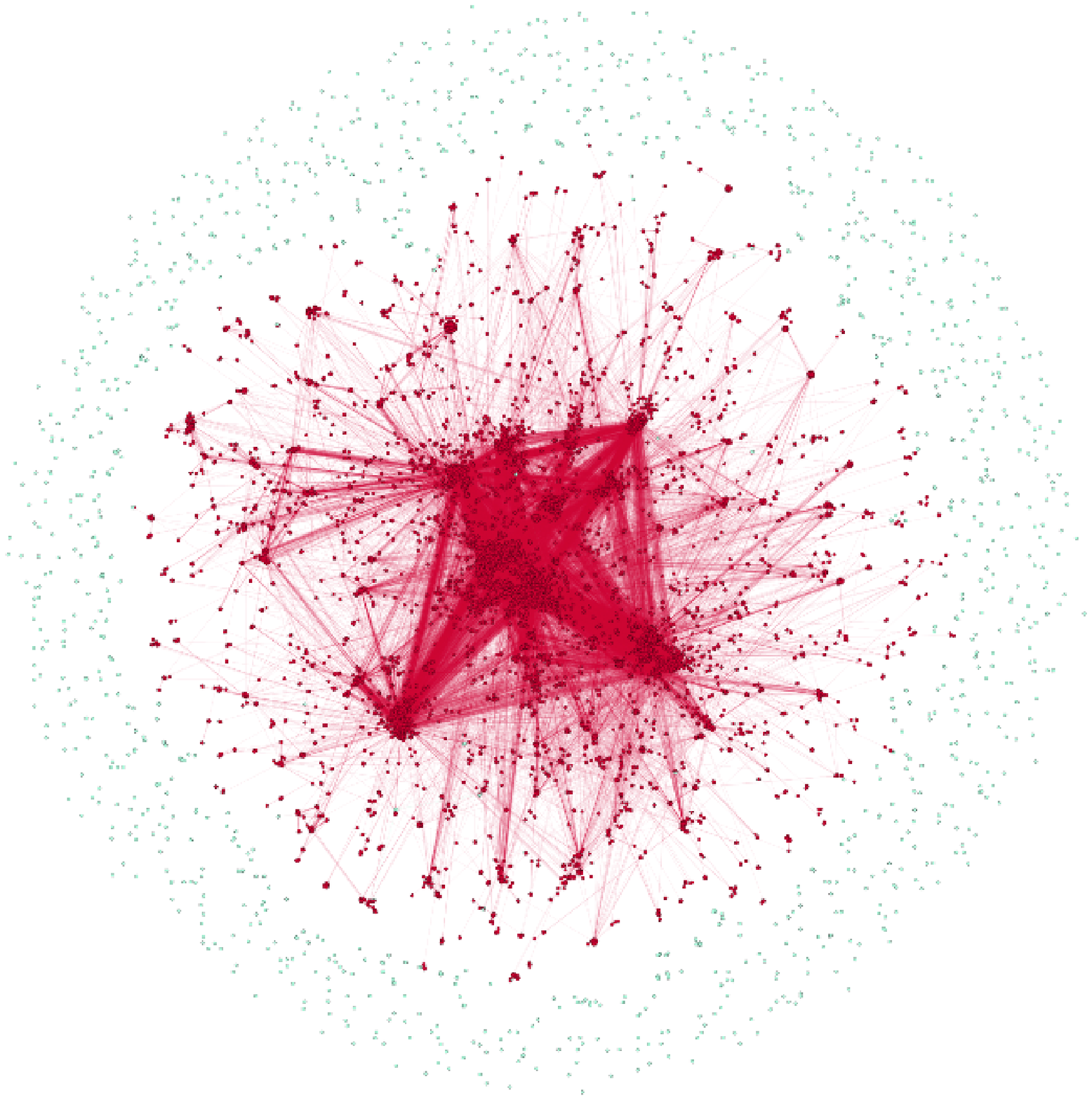,scale=0.3}
   \label{fig:fin_infected}
   \subcaption{Final Status}
\end{minipage}
\hfill
\caption{Visualization of the Enron email network infection process. In Figure (a) it is depicted the network after the infection initialization. Red nodes (with respective red edges) are infected while green nodes are susceptible. In Figure (b) it is depicted the final state where the network is flooded and the majority of nodes are infected. }
\label{emailenron_full}
\end{figure*}

Now we move to the case where profile $\mathcal{A}$ has high sensitivity against the virus while profile $\mathcal{B}$ presents low sensitivity against the virus. In this case, we expect that a small percentage of nodes in profile $\mathcal{B}$ will get infected whereas in profile $\mathcal{A}$ the majority of nodes will get infected. The parameter values used for this experiment are $(\beta_{\mathcal{A}},\delta_{\mathcal{A}} )=(0.006, 0.0001)$ and $(\beta_{\mathcal{B}},\delta_{\mathcal{B}} )=(0.009,0.1)$
 for $ (\alpha,b)=(0.001, 0.99) $. The results are depicted in Figure \ref{emailenron_1cN}.

\begin{figure}[ht]
\begin{center}
\epsfig{file=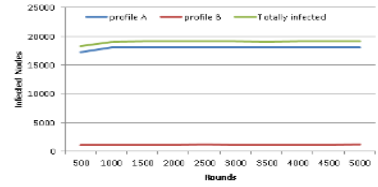,height=1in, width=2in}
\end{center}
\caption{Simulation in Enron email network where profile $\mathcal{A}$ has high sensitivity against the virus while profile $\mathcal{B}$ presents low sensitivity.}
 \label{emailenron_1cN}
\end{figure}

Last, we provide an example where we have randomly and equally divided the network in five profiles and we follow the same process like mentioned above (we infect the top 5\% of nodes according to their degree). The results are depicted in Figure \ref{emailenron_5}. In this case, we observe that profiles 0, 1, 2 are down to zero for all rounds of execution while profile 3 is steadily up to $\sim$ 580 infected nodes and profile 4 is up to $\sim$ 25580 nodes. Since all profiles present a steady behavior (they have the same amount of infected nodes for the majority of rounds), this can be considered as a stable fixed point. The fact that only two profiles from all five have infected nodes, is due to the fact that all top 5$\%$ nodes based on degree are, in those profiles (profile 3, profile 4). The healing and infection parameters used are: $(\beta_{\mathcal{0}},\delta_{\mathcal{0}} )=(0.006, 0.0001)$,$(\beta_{\mathcal{1}},\delta_{\mathcal{1}} )=(0.009,0.1)$,$(\beta_{\mathcal{2}},\delta_{\mathcal{2}} )=(0.006, 0.0001)$,$(\beta_{\mathcal{3}},\delta_{\mathcal{3}} )=(0.009,0.1)$ and $(\beta_{\mathcal{2}},\delta_{\mathcal{2}} )=(0.006, 0.0001)$
 for $ (\alpha,b)=(0.001, 0.99) $.

\begin{figure}[ht]
\begin{center}
\epsfig{file=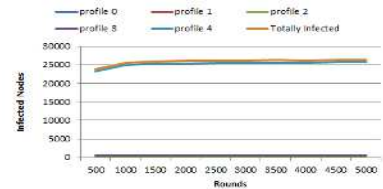,height=1in, width=2in}
\end{center}
\caption{Simulation in Enron email network where there are 5 profiles and profile 3 and profile 4 present  high sensitivity against the virus while profiles 0, 1, 2 present low sensitivity.}
 \label{emailenron_5}
\end{figure}

Most social networks, resemble powerlaw graphs. Consequently, we thought it would be wise to make a first attemp and have simulations in such graphs. A simple example, is depicted in Figure \ref{powerlaw}. The parameter values used in the creation of the graph are:$\alpha=2.72$ and $\beta=3000$. The infection and healing parameters used in the experiment are $(\beta_{\mathcal{A}},\delta_{\mathcal{A}} )=(0.5, 0.0001)$,$(\beta_{\mathcal{B}},\delta_{\mathcal{B}} )=(0.7,0.001)$. While the network is not flooded (only 530 nodes are infected from the total 5000 nodes), the system presents a steady behavior followed for the majority of execution rounds. In Figure \ref{powerlaw}, we present only the first 12 rounds of execution where the system converges up to 530 infected nodes. We omit the rest of rounds since the amount of infected nodes, remains the same and as a result, this is considered to be a stable fixed point.

\begin{figure}[ht]
\begin{center}
\epsfig{file=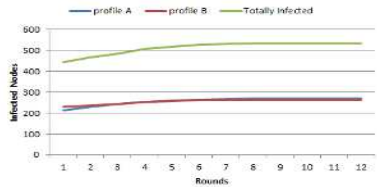,height=1in, width=2in}
\end{center}
\caption{Simulation in a powerlaw network where both profiles have high susceptibility against the virus but due to the structure of the graph, the virus propagation is limited.}
 \label{powerlaw}
\end{figure}

Finally, we present an alternative approach in profiling graphs, using the Montgomery network. Using the social contact graph file and the corresponding demographics file which are publicly available, we created a dataset file which consists of all interactions where besides all other characteristics, the age of the interacting nodes is included. Unfortunately, the available demographics dataset has missing values. For this reason, we omitted all interactions of nodes where the respective age of the node is not available. Consequently, the resulting graph consists of $67700$ nodes and $1626453$ edges. Using the age parameter, we divided nodes in five age groups: $1)$ where nodes are children below the age limit of 10 years, $2)$ where nodes correspond to children from the age of 10 till the age of 18, $3)$ where the nodes are adults in the range of (18-30) years, $4)$ where the nodes are adults in the range of (30-50) years and $5)$ where the nodes are adults with age $50+$.

With this profile initialization in the experiments, we infected 20 children in profile $1$ as a highly sensitive group against attacking virus and 1 in 100 humans in all other age profiles, resulting in 2942 initially infected nodes. The infection and healing parameter values that were used here are $\beta_{1}=\beta_{5}=0.9, \delta_{1}=\delta_{2}=\delta_{3}=\delta_{4}=\delta_{5}=0.01$ and $\beta_{2}=\beta_{3}=\beta_{4}=0.6$ for $(\alpha,b)=(0.001,0.99)$. The parameter values that were chosen, satisfy the assumption that children under the age of $10$ are a highly sensitive social group against the virus. The results are depicted in Figure~\ref{montgomery}.

\begin{figure}[h]
\begin{center}
\epsfig{file=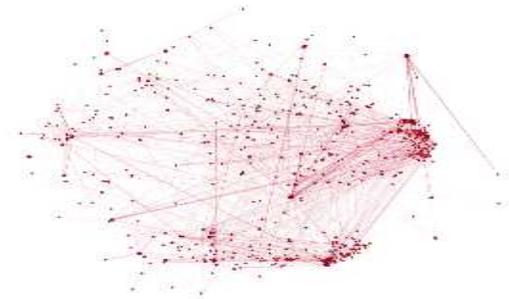,height=2in, width=3in}
\end{center}
\caption{Visualization of the Montgomery network infected nodes, after $5000$ simulation rounds using age profiles. The red nodes represent children under the age of $10$. The blue nodes represent adults in the age range $(18-30)$, the green nodes represent adults in the age range $(30-50)$ and yellow nodes represent adults with 50+ age.}
 \label{montgomery}
\end{figure}

One could notice that even though the parameter values that were chosen, were very high, the amount of infected nodes in Figure \ref{montgomery} is very small. This is due to the topology of the graph and the corresponding connectivity. While the maximum degree is $378$, most nodes have degree in the range $(1-10)$. In Figure \ref{montgomery}, we provide a visualization of the finally infected nodes after $5000$ rounds simulation.

\section{Discussion} \label{sec:discuss}

Here we discuss the results we presented in the previous sections and their possible extensions. The main characteristic of our setting is that there are infinite fixed points based on the relationship between the various parameters of the problem. This is in contrast to the finite and small number of fixed points in the case of two viruses \cite{Prakash:2012:WTC:2187836.2187975}. One could erroneously think that having two profiles in the network is like having two viruses but the truth is that the introduced heterogeneity of the underlying network adds complexity to the problem of finding the necessary stability conditions for the fixed points.

\paragraph{Clique} We have given conditions so that a virus in presence of two different profiles in the network will die out. We also gave conditions so that a particular number of nodes will get infected from each profile, thus connecting the footprint of the virus in the profiles with the parameters of the profiles with respect to the virus. Of course, we have provided such results for particular interesting cases since tackling the general case seems much harder. It is a matter of messy computations to do the same for a barbell graph with uniform weights on the edges between its two cliques. However, adding arbitrary weights on the barbell graph or even on the clique requires a more general approach where the characteristics of the adjacency matrix must be taken into account. 

\paragraph{Arbitrary Graph} 
In this case, we have provided a general condition in \ref{sssec:zero} so that the virus will die out or persist in the network. In case the virus persist, we prove conditions that should hold for the graph so that the steady-state infection probability for each node is within some prespecified range. Although this is a useful result, it is not the whole story. This is because we impose that the probabilities of \textbf{all} nodes should be within this range. As a result, we fail to catch the case where most of the nodes are within this range but there are some nodes with probabilities that are outside this range. For example, imagine a clique $K_n$ and a path $P_n$ of $n$ nodes respectively so that the path $P_n$ hangs from some node in $K_n$ creating a graph of $2n$ nodes in total. It is expected that nodes in $P_n$ will have lower probabilities than those in $K_n$ and thus some of them may be out of the prespecified range. To tackle these cases one needs to fully solve the respective dynamical system.

\paragraph{Profiling}
How does one specify the profile of a node in a given network that captures the relationships between agents within a particular framework? Take for example an epidemiological scenario where the virus is the flu. The network specifies the contact between people during a day. It is known that there are groups that are more susceptible to the virus than other groups of people (e.g., children and adults). In this case, one would propose to specify profiles based on the age of nodes (as we have done in one of our experiments).  An interesting approach is presented in \cite{aral} where profiling in a social experiment, shed light in potentially influential users. In a social network scenario, one could also specify the affinity towards a particular rumor or idea (e.g., a PS4 game) by looking at relative historical data of each agent and then decide whether each agent is more susceptible or less susceptible to this particular rumor or idea (or even class of rumors and ideas). However, there is still the problem of giving a value that describes the affinity of each agent. This can either be the choice of the researcher or can be accomplished by using a classifier working on relative historical data, if there is such data of course. Summarizing, we feel that an empirical study of such an extent would be very interesting and it would be a different and surely an interesting paper.

\section{Conclusions} \label{conclusion}
In this paper, we studied the case where one competing virus/rumor/product is spreading over a heterogeneous network. In this network, the nodes have different endurance against the "virus" and we answer the question of what will happen in the end by providing the necessary conditions so that the system will reach a steady state. We proved for different scenarios, the fixed points the system can reach and the stability conditions that are required. Our main results concern the clique and arbitrary topologies. We also verified the theoretical analysis with simulation experiments on synthetic and real-world datasets. 

Future directions include the extension of this work to other virus propagation models as well as the study of multiple profiles and multiple viruses on a single network. The theoretical analysis of such a case may be very difficult using tools from dynamical systems theory but we feel that a more algorithmic approach may bear fruits (algorithmic analytical tools for Natural Algorithms e.g., \cite{conf:innovations:Chazelle10}). A more extensive experimental evaluation will be included in the journal version.

\section{Acknowledgments}
The source code for the experiments can be found at \url{http://alkistis.ceid.upatras.gr/research/kdd15/}.
%
\bibliographystyle{abbrv}
\bibliography{virus}
%
%
\appendix
\section{Proof of Positive Discriminant for Fixed Point $I_{\mathcal{A}},I_{\mathcal{B}}\rightarrow 0$ in a Clique} \label{app:pos_discr}
The discriminant is:
\[\Delta=(\delta_{\mathcal{A}}+\delta_{\mathcal{B}}-N(\beta_{\mathcal{A}}+\beta_{\mathcal{B}}))^2-4\delta_{\mathcal{A}}\delta_{\mathcal{B}}+4N(\delta_{\mathcal{B}}\beta_{\mathcal{A}}+\delta_{\mathcal{A}}\beta_{\mathcal{B}})\]
which results in:

\noindent $\Delta=\delta^2_{\mathcal{A}}+\delta^2_{\mathcal{B}}-2\delta_{\mathcal{A}}\delta_{\mathcal{B}}+N^2(\beta_{\mathcal{A}}+\beta_{\mathcal{B}})^2-2N(\delta_{\mathcal{A}}+\delta_{\mathcal{B}})(\beta_{\mathcal{A}}+\beta_{\mathcal{B}})+4N(\delta_{\mathcal{B}}\beta_{\mathcal{A}}+\delta_{\mathcal{A}}\beta_{\mathcal{B}})$

The first three terms can be written equivalently either as $(\delta_{\mathcal{A}}-\delta_{\mathcal{B}})^2$ or $(\delta_{\mathcal{B}}-\delta_{\mathcal{A}})^2$. Assuming that $\delta_{\mathcal{B}}>\delta_{\mathcal{A}}$ we get that the discriminant can be written as:
\[\Delta=((\delta_{\mathcal{B}}-\delta_{\mathcal{A}})-N(\beta_{\mathcal{A}}+\beta_{\mathcal{B}}))^2+4N\beta_{\mathcal{A}}(\delta_{\mathcal{B}}-\delta_{\mathcal{A}})>0
\]
and for the symmetric case where $\delta_{\mathcal{A}}>\delta_{\mathcal{B}}$ we get an equivalent formula for the discriminant which clearly shows that it is also always positive.
\[\Delta=((\delta_{\mathcal{A}}-\delta_{\mathcal{B}})-N(\beta_{\mathcal{A}}+\beta_{\mathcal{B}}))^2+4N\beta_{\mathcal{B}}(\delta_{\mathcal{A}}-\delta_{\mathcal{B}})>0
\]
Finally, if $\delta_{\mathcal{A}}=\delta_{\mathcal{B}}$ the discriminant is also positive and equal to $N^2(\beta_{\mathcal{A}}+\beta_{\mathcal{B}})^2$.

\end{document}